\pdfoutput=1
\documentclass[12pt]{article}
\usepackage[hmargin=1in,top=1in,bottom=1.25in]{geometry}

\usepackage{style/qstyle}
\usepackage[normalem]{ulem}

\renewcommand{\emph}[1]{\textbf{#1}}
\renewcommand{\mathbf}[1]{\textit{\textbf{#1}}}

\AddToHook{env/theorem/begin}{\crefalias{definition}{theorem}}
\AddToHook{env/lemma/begin}{\crefalias{definition}{lemma}}
\AddToHook{env/corollary/begin}{\crefalias{definition}{corollary}}
\AddToHook{env/proposition/begin}{\crefalias{definition}{proposition}}
\AddToHook{env/conjecture/begin}{\crefalias{definition}{conjecture}}
\AddToHook{env/example/begin}{\crefalias{definition}{example}}

\crefname{theorem}{Theorem}{Theorems}
\crefname{lemma}{Lemma}{Lemmata}
\crefname{corollary}{Corollary}{Corollaries}
\crefname{proposition}{Proposition}{Propositions}
\crefname{conjecture}{Conjecture}{Conjectures}
\crefname{example}{Example}{Examples}
\crefname{definition}{Definition}{Definitions}

\title{Statistically secure uncloneable encryption\\of arbitrary messages}

\author[1]{Archishna Bhattacharyya\footnote{abhat086@uottawa.ca}}
\author[1]{Anne Broadbent\footnote{abroadbe@uottawa.ca}}
\author[2]{Eric Culf\footnote{eculf@uwaterloo.ca}}
\affil[1]{University of Ottawa}
\affil[2]{Institute for Quantum Computing, University of Waterloo}

\date{}

\begin{document}

\maketitle

\begin{abstract}
	Unconditional uncloneable encryption of a single bit with efficient encryption and decryption is now possible. However, whether the extension to messages of arbitrary length achieves statistical security remains to be known. Using the fact that the encoding bases for the single-bit scheme known to be secure consist of a subset of the Clifford unitaries, we show that this scheme can be upgraded to achieve unconditional uncloneable encryption for messages of arbitrary length, with encoding time polynomial in the message length and security parameter. This establishes that one-time uncloneable encryption of arbitrary messages enjoys statistical security.\footnote{This work was performed without the aid of artificial intelligence tools.}
\end{abstract}

\vspace{1cm}

\tableofcontents

\newpage

\section{Introduction}

In this work, we resolve the problem of uncloneable encryption of messages with arbitrary length, and show that it is statistically secure. We expand on our earlier work \cite{BBC26} showing unconditional security of the uncloneable bit, and establish unconditional security for arbitrary messages with efficient encryption and decryption. The scheme we work with is a multi-message generalisation of the Clifford encoding studied in~\cite{BC26}. 

Recently, unconditional security of the uncloneable bit with efficient encryption and decryption was achieved in \cite{AS26, Rag26}. In the latter, uncloneable encryption of messages with arbitrary length permitting many uses was shown to be secure under computational assumptions. While reusable uncloneable encryption can only be computationally secure, we show that it is possible to have unconditional security of single-use uncloneable encryption with messages of arbitrary length, and encoding time polynomial in the message length and security parameter. This was notably left open in \cite{Rag26} where the construction relies on pseudorandom-function like states, and prior to that \cite{HKNY24} showed that the extension from a single bit to strings is possible with one-way functions.

The pursuit of uncloneable encryption has been an intriguing and beautiful venture of modern cryptography and concludes with information-theoretic security of the primitive. We do not aim to introduce or summarise that here, instead we simply exposit a short analysis that completes this pursuit in its full strength. The reader may refer to the original construction and conjecture by Broadbent and Lord \cite{BL20}, and subsequent works thereafter which have since explored varied routes to security, and applications~\cite{MST21arxiv,ALP21,CLLZ21,AK21,AKL+22,HKNY24,AKL23,BC23arxiv,AKY25,AB24,JK24}. The foundation built in \cite{BC26, BBC26} is an insightful account of why uncloneable encryption, to much surprise, achieved information-theoretic security, when it was not even certain whether security could be established with computational assumptions.

At the onset of our analysis is the idea that one can exploit structural properties of the Clifford group to extend the unconditional security of the uncloneable bit \cite{BBC26}, with efficient encryption and decryption \cite{AS26}, to the encryption of multiple messages. Precisely, the scheme originally studied in \cite{BC26}, which strengthens the scheme of~\cite{BBCNPR26}, is now unconditionally secure \cite{AS26, Rag26}. Using the fact that this scheme consists of a subset of the Clifford unitaries, we modify the scheme in a relatively simple manner to establish our claim. Then, we present a refinement of the techniques in \cite{AS26}, using which we establish the optimal security bound possible for messages of arbitrary length, as reported in \cite{BCR25}.

\paragraph{Acknowledgements} A.Bh. and A.Br. acknowledge the support of the Natural Sciences and Engineering Research Council of Canada (NSERC)(ALLRP-578455-2022, RGPIN-2022-05167), the Air Force Office of Scientific Research under award number FA9550-20-1-0375 and of the Canada Research Chairs Program (CRC-2023-00173). E.C. is supported by a \mbox{CGS D} scholarship from NSERC.

\section{Preliminaries}

\begin{definition} \label{def:qecm}
    A \emph{quantum encryption of classical messages (QECM)} is given by a tuple $\ttt{Q}=(K,X,A,\mu,\{\sigma^k_x\}_{k\in K,x\in X})$, where
    \begin{itemize}
        \item $K$ is a set, representing the encryption keys;
        \item $X$ is a finite set, representing the messages;
        \item $A$ is a register, representing the system holding the encrypted messages;
        \item $\mu$ is a probability measure on $K$, representing the key distribution;
        \item $\sigma^k_x\in D(\mc{H}_A)$ is a quantum state, representing the encryption of message $x$ with key $k$.
    \end{itemize}
    We say a QECM is \emph{$\eta$-correct} if there exists a family of CPTP maps $\Phi^k:B(\mc{H}_A)\rightarrow B(\mc{H}_X)$, called decryption maps, such that for all $k\in K$ and $x\in X$,
    \begin{align*}
        \braket{x}{\Phi^k(\sigma^k_x)}{x}\geq\eta.
    \end{align*}
    We say that the QECM is \emph{correct} if it is $1$-correct.
    
    We say a family of QECMs $\{\ttt{Q}_\lambda\}_{\lambda\in\N}$ is an \emph{efficient QECM} if key sampling, encrypted message preparation, and decryption can be implemented in polynomial time in $\lambda$.
\end{definition}

Note that correctness is equivalent to the orthogonality condition $\Tr(\sigma^k_{x}\sigma^k_{x'})=0$ for $k\in K$ and $x\neq x'\in X$.

\begin{definition} \label{def:cloning-att}
    A \emph{cloning attack} against a QECM $\ttt{Q}=(K,X,A,\mu,\{\sigma^k_x\}_{k\in K,x\in X})$ is a tuple $\ttt{A}=(B,C,\{B^k_x\}_{k\in K,x\in X},\{C^k_x\}_{k\in K,x\in X},\Phi)$, where
    \begin{itemize}
        \item $B$ and $C$ are registers, representing Bob and Charlie's systems, respectively;
        \item $\{B^k_x\}_{x\in X}\subseteq B(\mc{H}_B)$ and $\{C^k_x\}_{x\in X}\subseteq B(\mc{H}_C)$ are POVMs, representing Bob and Charlie's measurements given key $k$, respectively;
        \item $\Phi:B(\mc{H}_A)\rightarrow B(\mc{H}_{BC})$ is a CPTP map, representing the cloning channel.
    \end{itemize}
    The \emph{success probability} of $\ttt{A}$ against $\ttt{Q}$ is
    \begin{align}
        \mfk{c}(\ttt{Q},\ttt{A})=\expec_{k\leftarrow \mu}\frac{1}{|X|}\sum_{x\in X}\Tr\squ*{(B^k_x\otimes C^k_x)\Phi(\sigma^k_x)}.
    \end{align}
    The \emph{cloning value} of $\ttt{Q}$ is $\mfk{c}(\ttt{Q})=\sup_{\ttt{A}}\mfk{c}(\ttt{Q},\ttt{A})$, where the supremum is over all cloning attacks. We say a QECM is \emph{$\delta$-uncloneable secure} if $\mfk{c}(\ttt{Q})\leq\frac{1}{|X|}+\delta$.

    For a function $f:\N\rightarrow[0,1]$, we say a family of QECMs $\{\ttt{Q}_\lambda\}$ is \emph{$f$-uncloneable secure} if $\ttt{Q}_\lambda$ is $f(\lambda)$-uncloneable secure for all $\lambda$. We additionally say $\{\ttt{Q}_\lambda\}$ is \emph{uncloneable secure} if $\lim\limits_{\lambda\rightarrow\infty}f(\lambda)=0$; and $\{\ttt{Q}_\lambda\}$ is \emph{strongly uncloneable secure} if $f$ is a negligible function.
\end{definition}

\begin{definition}\label{def:cloning-distinguishing-att}
    A \emph{cloning-distinguishing attack} against a QECM $\ttt{Q}=(K,X,A,\mu,\{\sigma^k_x\}_{k\in K,x\in X})$ is a tuple $\ttt{A}=(\{x_0,x_1\},B,C,\{B^k_b\}_{k\in K,b\in\{0,1\}},\{C^k_b\}_{k\in K,b\in\{0,1\}},\Phi)$, where
    \begin{itemize}
        \item $x_0\neq x_1\in X$ are distinct messages, representing the two messages to be distinguished;
        \item $B$ and $C$ are registers, representing Bob and Charlie's systems, respectively;
        \item $\{B^k_b\}_{b\in\{0,1\}}\subseteq B(\mc{H}_B)$ and $\{C^k_b\}_{b\in\{0,1\}}\subseteq B(\mc{H}_C)$ are POVMs, representing Bob and Charlie's measurements given key $k$, respectively;
        \item $\Phi:B(\mc{H}_A)\rightarrow B(\mc{H}_{BC})$ is a CPTP map, representing the cloning channel.
    \end{itemize}
    The \emph{success probability} of $\ttt{A}$ against $\ttt{Q}$ is
    \begin{align}
        \mfk{cd}(\ttt{Q},\ttt{A})=\expec_{k\leftarrow\mu}\frac{1}{2}\sum_{b\in\{0,1\}}\Tr\squ*{(B^k_b\otimes C^k_b)\Phi(\sigma^k_{x_b})}.
    \end{align}
    The \emph{cloning-distinguishing value} of $\ttt{Q}$ is $\mfk{cd}(\ttt{Q})=\sup_{\ttt{A}}\mfk{cd}(\ttt{Q},\ttt{A})$, where the supremum is over all cloning-distinguishing attacks. We say a QECM is \emph{$\delta$-uncloneable-indistinguishable secure} if~$\mfk{cd}(\ttt{Q})\leq\frac{1}{2}+\delta$.

    For a function $f:\N\rightarrow[0,1]$, we say a family of QECMs $\{\ttt{Q}_\lambda\}$ is \emph{$f$-uncloneable-in\-dis\-ting\-ui\-sha\-ble secure} if $\ttt{Q}_\lambda$ is $f(\lambda)$-uncloneable-indistiguishable secure for all $\lambda$. We additionally say~$\{\ttt{Q}_\lambda\}$ is \emph{uncloneable-indistinguishable secure} if $\lim\limits_{\lambda\rightarrow\infty}f(\lambda)=0$; and $\{\ttt{Q}_\lambda\}$ is \emph{strongly uncloneable-indistinguishable secure} if $f$ is a negligible function.
\end{definition}

It is often more natural to study monogamy-of-entanglement games instead of QECMs.

\begin{definition}
    A \emph{monogamy-of-entanglement (MoE) game} is a tuple $\ttt{G}=(\Theta,X,A,\mu,\{A^\theta_x\}_{\theta\in\Theta,x\in X})$, where
    \begin{itemize}
        \item $\Theta$ is a set, representing the questions;
        \item $X$ is a finite set, representing the answers;
        \item $A$ is a register, representing Alice's system;
        \item $\mu$ is a probability measure on $\Theta$, representing the question distribution.
        \item $\{A^\theta_x\}_{x\in X}\subseteq B(\mc{H}_A)$ is a POVM, representing Alice's measurements given question $\theta$.
    \end{itemize}
    A \emph{strategy} for an MoE game $\ttt{G}$ is a tuple $\ttt{S}=(B,C,\{B^\theta_x\}_{\theta\in\Theta,x\in X},\{C^\theta_x\}_{\theta\in \Theta,x\in X},\rho_{ABC})$, where
    \begin{itemize}
        \item $B$ and $C$ are registers, representing Bob and Charlie's systems, respectively;
        \item $\{B^\theta_x\}_{x\in X}\subseteq B(\mc{H}_B)$ and $\{C^\theta_x\}_{x\in X}\subseteq B(\mc{H}_C)$ are POVMs, representing Bob and Charlie's measurements given question $\theta$, respectively;
        \item $\rho_{ABC}\in D(\mc{H}_{ABC})$ is the shared quantum state.
    \end{itemize}
    The \emph{winning probability} of $\ttt{S}$ at $\ttt{G}$ is
    \begin{align}
        \mfk{w}(\ttt{G},\ttt{S})=\expec_{\theta\leftarrow\mu}\sum_{x\in X}\Tr\squ*{(A^\theta_x\otimes B^\theta_x\otimes C^\theta_x)\rho_{ABC}}.
    \end{align}
    The \emph{quantum value} of $\ttt{G}$ is $\mfk{w}(\ttt{G})=\sup_{\ttt{S}}\mfk{w}(\ttt{G},\ttt{S})$, where the supremum is over all strategies.
\end{definition}

Now, we introduce the QECM scheme we study in this paper.

\begin{definition}
	Let $n\geq m$. Let $A=\{0,1\}^n$, and for $x\in\{0,1\}^m$, let $\sigma_x=\frac{1}{2^{n-m}}\ketbra{x}\otimes I$. Let $\mc{V}\subseteq U(\mc{H}_A)$ be a finite set, and write $\mu_{\mc{V}}$ for the uniform distribution on $\mc{V}$. Define the QECM $\ttt{Q}_{\mc{V},m}=(\mc{V},\{0,1\}^m,A,\mu_{\mc{V}},\{U\sigma_xU^\dag\}_{U,x})$.
\end{definition}

Write $\mc{C}_n$ for the Clifford group on $n$ qubits. This gives rise to QECMs we call \emph{Clifford schemes}. Then, for any polynomials $p(\lambda)$ and $q(\lambda)$, the Clifford scheme $\{\ttt{Q}_{\mc{C}_{p(\lambda)},q(\lambda)}\}_\lambda$ is an efficient QECM with message length~$q(\lambda)$.

Each of these QECMs has an equivalent MoE game.

\begin{definition}
	Let $n\geq m$. Let $A=\{0,1\}^n$, and for $x\in\{0,1\}^m$, let $\Pi_x=\ketbra{x}\otimes I$. Let $\mc{V}\subseteq U(\mc{H}_A)$ be a finite set. Define the MoE game $\ttt{Q}_{\mc{V},m}=(\mc{V},\{0,1\}^m,A,\mu_{\mc{V}},\{U\Pi_xU^\dag\}_{U,x})$.
\end{definition}

The winning probabilities are related as $\mfk{c}(\ttt{Q}_{\mc{V},m})\leq\mfk{w}(\ttt{G}_{\mc{V},m})$ and $\mfk{c}(\ttt{Q}_{\mc{V},1})=\mfk{cd}(\ttt{Q}_{\mc{V},1})$.

\section{Extending message length}

In this section, we relate the cloning-distinguishing value of the Clifford scheme with $m$-bit messages defined above, to the cloning-distinguishing value of a Clifford scheme with $1$-bit messages, whose value may be bounded using the techniques of \cite{AS26,Rag26}. To show this relation, we rely on the group structure of the Clifford group. Since the question distribution is uniform, it corresponds to the Haar measure on the Clifford group, and hence it is invariant under the action of the group on itself. Using this, we can translate a cloning-distinguishing attack against the $m$-bit Clifford scheme to a cloning-distinguishing attack for a $1$-bit scheme for a subgroup of the Clifford group, isomorphic to the Clifford group on fewer qubits.

\begin{theorem}\label{thm:one-to-many}
	The cloning-distinguishing value $\mfk{cd}(\ttt{Q}_{\mc{C}_n,m})\leq\mfk{cd}(\ttt{Q}_{\mc{C}_{n-m+1},1})$.
\end{theorem}

\begin{proof}
	Let $\ttt{A}=(\{x_0,x_1\},B,C,\{B^U_b\},\{C^U_b\},\Phi)$ be a cloning-distinguishing attack. Since there is a Clifford $C$ such that $C\sigma_{x_0}C^\dag=\sigma_{0^m}$ and $C\sigma_{x_1}C^\dag=\sigma_{0^{m-1}1}$, we may assume $x_0=0^m$ and $x_1=0^{m-1}1$. Now, for each $V\in\mc{C}_n$, define the following cloning-distinguishing attack against $\ttt{Q}_{\mc{C}_{n-m+1},1}$: $\ttt{A}_V=(\{0,1\},B,C,\{B^{V(I\otimes U)}_b\}_{U,b},\{C^{V(I\otimes U)}_b\}_{U,b},\Phi_V)$, where $\Phi_V(\rho)=\Phi(V(\ketbra{0}^{m-1}\otimes\rho)V^\dag)$. Now, using the Haar invariance of the uniform measure on a finite group,
	\begin{align*}
		\mfk{cd}(\ttt{Q}_{\mc{C}_n,m},\ttt{A})&=\expec_{V\in\mc{C}_n}\frac{1}{2}\sum_{b\in\{0,1\}}\Tr\squ*{(B^V_b\otimes C^V_b)\Phi(V\sigma_{x_b}V^\dag)}\\
		&=\expec_{U\in\mc{C}_{n-m+1}}\expec_{V\in\mc{C}_n}\int_{\mc{C}_n}\frac{1}{2}\sum_{b\in\{0,1\}}\Tr\squ*{(B^{V(I\otimes U)}_b\otimes C^{V(I\otimes U)}_b)\Phi(V(I\otimes U)\sigma_{x_b}(I\otimes U^\dag)V^\dag)}\\
		&=\expec_{U\in\mc{C}_{n-m+1}}\expec_{V\in\mc{C}_n}\frac{1}{2}\sum_{b\in\{0,1\}}\Tr\squ*{(B^{V(I\otimes U)}_b\otimes C^{V(I\otimes U)}_b)\Phi_V(U\sigma_b U^\dag)}\\
		&=\expec_{V\in\mc{C}_n}\mfk{cd}(\ttt{Q}_{\mc{C}_{n-m+1},1},\ttt{A}_V)\\
		&\leq \mfk{cd}(\ttt{Q}_{\mc{C}_{n-m+1},1}).
	\end{align*}
	Taking the supremum over strategies $\ttt{A}$ gives the wanted result.
\end{proof}

The same result also holds for the Haar measure game, because of the properties that the question distribution is the Haar measure over a group, and $I\otimes U\in\mc{U}(2^n)$ for all $U\in\mc{U}(2^m)$.

\section{General upper bound}

In this section, we show a slight generalisation of the main result of~\cite{AS26}, with a streamlined proof.

\begin{definition}
	Let $\mc{A}=\{A_\theta\}_{\theta\in\Theta}$ be a collection of binary observables and let $\mu$ be a probability distribution on $\Theta$. The \emph{observable overlap} is
	\begin{align}
		c(\mc{A},\mu)=\expec_{\theta,\theta'\leftarrow\mu}\abs*{\Tr(A_\theta A_{\theta'})}.
	\end{align}
\end{definition}

The overlap $c(\mc{A},\mu)$ is always upper-bounded by the dimension $d$. For unitary observables and a fixed number of questions $|\Theta|$, the overlap is minimised when the observables are orthogonal with respect to the Hilbert-Schmidt inner product, in which case it becomes
\begin{align*}
	c(\mc{A},\mu)=\sum_{\theta\in\Theta}\mu(\theta)^2\Tr(A_\theta^2)=d\sum_{\theta\in\Theta}\mu(\theta)^2=2^{\log d-H_2(\mu)}.
\end{align*}
This is minimised when the distribution is uniform, where $c(\mc{A},\mu)=\frac{d}{|\Theta|}$. Note that in dimension $d$, $|\Theta|\leq d^2$, where equality is attained for $d=2^n$ via the orthogonal basis of Pauli operators.

The observable overlap is similar in spirit to the maximal overlap of measurements studied in~\cite{TFKW13}, which provided the first upper bound on the cloning value of a QECM scheme.

\begin{theorem}\label{thm:ananth-sahai}
	Let $\ttt{G}=(\Theta,\{0,1\},A,\mu,\{A^\theta_b\}_{\theta,b})$ be a MoE game with single-bit output where Alice's observables $A_\theta=A^\theta_0-A^\theta_1$ are unitary. Then, the
	\begin{align*}
		\mfk{w}(\ttt{G})\leq\frac{1}{2}+\frac{1}{2}\sqrt{c(\{A_\theta\}_\theta,\mu)}.
	\end{align*}
\end{theorem}

\begin{lemma}[Conditional overlap \cite{AS26}]\label{lem:ool}
	Let $U_{AB}$ and $V_{AC}$ be operators. Then, $$\norm{U_{AB}V_{AC}}\leq\sqrt{\norm{\Tr_A(U_{AB}^\dag U_{AB})}\norm{\Tr_A(V_{AC}V_{AC}^\dag)}}.$$
\end{lemma}

\begin{proof}
	We expand $U_{AB}=\sum_{i,j}\ketbra{i}{j}_A\otimes U^{i,j}_B\otimes I_C$ and $V_{AC}=\sum_{j,k}\ketbra{j}{k}_A\otimes I_B\otimes V^{j,k}_C$. Thus, $U_{AB}V_{AC}=\sum_{i,j,k}\ketbra{i}{k}\otimes U^{i,j}_B\otimes V^{j,k}_C$. Next, fix arbitrary unit vectors $\ket{\psi},\ket{\phi}$. Let $\ket{\psi_{j,k}}=\sum_{i}(\ketbra{i}{k}\otimes U^{i,j}_B\otimes I_C)\ket{\psi}$ and $\ket{\phi_{j,k}}=(I_{AB}\otimes (V^{j,k}_C)^\dag)\ket{\phi}$. Then, using the Cauchy-Schwarz inequality,
	\begin{align*}
		\abs*{\braket{\phi}{U_{AB}V_{AC}}{\psi}}^2&=\abs[\Big]{\sum_{j,k}\braket{\phi_{j,k}}{\psi_{j,k}}}^2\leq\sum_{j,k}\braket{\psi_{j,k}}\sum_{j,k}\braket{\phi_{j,k}}.
	\end{align*}
	We can bound the two terms in the product individually. First,
	$$\sum_{j,k}\braket{\psi_{j,k}}=\sum_{i,j,k}\braket{\psi}{\ketbra{k}\otimes (U^{i,j}_B)^\dag U^{i,j}_B\otimes I_C}{\psi}\leq\norm[\Big]{\sum_{i,j}(U^{i,j}_B)^\dag U^{i,j}_B}=\norm{\Tr_A(U_{AB}^\dag U_{AB})},$$
	and similarly $\sum_{j,k}\braket{\phi_{j,k}}=\sum_{j,k}\braket{\phi}{I_A\otimes I_B\otimes V^{j,k}_C(V^{j,k}_C)^\dag}{\phi}\leq\norm{\Tr_A(V_{AC} V_{AC}^\dag)}.$ To finish, note that the operator norm of $U_{AB}V_{AC}$ is attained as the supremum of $\abs*{\braket{\phi}{U_{AB}V_{AC}}{\psi}}$ over all pairs of unit vectors. 
\end{proof}

\begin{proof}[Proof of \cref{thm:ananth-sahai}]
	Let $\ttt{S}=(\mc{H}_B,\mc{H}_C,\{B^\theta_b\}_{\theta,b},\{C^\theta_b\}_{\theta,b},\ketbra{\psi})$ be a strategy for $\ttt{G}$ (we may without loss of generality assume that the measurements are projective and the shared state is pure). We suppose that the winning probability is $>\frac{1}{2}$ (otherwise the result clearly holds). Then, the winning probability in observable form is
	\begin{align*}
		\mfk{w}(\ttt{G},\ttt{S})=\frac{1}{4}+\frac{1}{4}\expec_{\theta\leftarrow\mu}\braket{\psi}{A_\theta\otimes B_\theta\otimes I+A_\theta\otimes I\otimes C_\theta+I\otimes B_\theta\otimes C_\theta}{\psi}.
	\end{align*}
	Write $\mrm{B}=\expec_{\theta\leftarrow\mu}A_\theta\otimes B_\theta\otimes I$, $\Gamma=\expec_{\theta\leftarrow\mu}A_\theta\otimes I\otimes C_\theta$, and $\Delta=\expec_{\theta\leftarrow\mu}I\otimes B_\theta\otimes C_\theta$. By construction $-I\leq \mrm{B},\Gamma,\Delta\leq I$. We may without loss of generality suppose that $\ket{\psi}$ is an eigenvector of $\mrm{B}+\Gamma+\Delta$ with maximal eigenvalue. Write $\mfk{w}(\ttt{G},\ttt{S})=\frac{1}{2}+\lambda$ for the winning probability. Then, we have $(\frac{1}{2}+\lambda)\ket{\psi}=\parens*{\frac{1}{4}+\frac{1}{4}(\mrm{B}+\Gamma+\Delta)}\ket{\psi}$, or, rearranging, $$(\mrm{B}+\Gamma)\ket{\psi}=\parens*{4\lambda+(I-\Delta)}\ket{\psi}.$$
	Since $\parens*{4\lambda+(I-\Delta)}$ is positive-definite, we can define $\mrm{E}=\parens*{4\lambda+(I-\Delta)}^{-\frac{1}{2}}$. Consider $\eta\coloneqq\norm{\mrm{E}(\mrm{B}+\Gamma)\ket{\psi}}^2-\norm{\mrm{E}(\mrm{B}-\Gamma)\ket{\psi}}^2$. First, expanding, we see that
	\begin{align*}
		\eta&=\braket{\psi}{(\mrm{B}+\Gamma)\mrm{E}^2(\mrm{B}+\Gamma)}{\psi}-\braket{\psi}{(\mrm{B}-\Gamma)\mrm{E}^2(\mrm{B}-\Gamma)}{\psi}\\
		&=2\braket{\psi}{\mrm{B}\mrm{E}^2\Gamma}{\psi}+2\braket{\psi}{\Gamma\mrm{E}^2\mrm{B}}{\psi}\\
		&=4\latRe\braket{\psi}{\mrm{B}\mrm{E}^2\Gamma}{\psi}.
	\end{align*}
	On the other hand, $\norm{\mrm{E}(\mrm{B}+\Gamma)\ket{\psi}}^2=\norm*{\parens{4\lambda+(I-\Delta)}^{\frac{1}{2}}\ket{\psi}}^2=4\lambda+1-\braket{\psi}{\Delta}{\psi}$. To bound the second term, we claim that $(\mrm{B}-\Gamma)\mrm{E}^2(\mrm{B}-\Gamma)\leq I-\Delta$. In the following, we assume that $\mrm{B}-\Gamma$ and $I-\Delta$ are invertible; since the invertible matrices are dense, we can approximate the matrices by invertible matrices to arbitrary precision. First,
	\begin{align*}
		I-\Delta\pm(\mrm{B}-\Gamma)&=\expec_{\theta\leftarrow\mu}I\pm(A_\theta\otimes B_\theta\otimes I+A_\theta\otimes I\otimes C_\theta)+I\otimes B_\theta\otimes C_\theta\\
		&=\expec_{\theta\leftarrow\mu}(I\pm A_\theta\otimes B_\theta\otimes I)(I\pm A_\theta\otimes I\otimes C_\theta)\geq 0,
	\end{align*}
	as each term is the product of commuting projections.
	This implies that $-I\leq(I-\Delta)^{-1/2}(\mrm{B}-\Gamma)(I-\Delta)^{-1/2}\leq I$. As all the eigenvalues are contained between $-1$ and $1$, taking the square gives
	$(I-\Delta)^{-1/2}(\mrm{B}-\Gamma)(I-\Delta)^{-1}(\mrm{B}-\Gamma)(I-\Delta)^{-1/2}\leq I.$
	Rearranging, we get that $$(\mrm{B}-\Gamma)(I-\Delta)^{-1}(\mrm{B}-\Gamma)\leq I-\Delta\leq 4\lambda+I-\Delta=\mrm{E}^{-2}.$$ Inverting, $(\mrm{B}-\Gamma)^{-1}(I-\Delta)(\mrm{B}-\Gamma)^{-1}\geq\mrm{E}^2$, which immediately implies $I-\Delta\geq(\mrm{B}-\Gamma)\mrm{E}^2(\mrm{B}-\Gamma)$.
	This implies that
	\begin{align*}
		\eta\geq4\lambda+1-\braket{\psi}{\Delta}{\psi}-\braket{\psi}{I-\Delta}{\psi}=4\lambda.
	\end{align*}
	Using the two expressions for $\eta$, $\lambda\leq\latRe\braket{\psi}{\mrm{B}\mrm{E}^2\Gamma}{\psi}\leq\norm{\mrm{B}\mrm{E}^2\Gamma}$. It remains to bound this overlap. First, we can see that this can be controlled by the simpler overlap $\norm{\mrm{B}\Gamma}$. To do so, note that $1+4\lambda>1$, which implies that $\frac{\Delta}{1+4\lambda}$ has its eigenvalues contained in the interval $(-1,1)$. As such, $\mrm{E}^2$ admits the power series expansion
	\begin{align*}
		\mrm{E}^2=(4\lambda+1-\Delta)^{-1}=\frac{1}{1+4\lambda}\parens*{1-\frac{\Delta}{1+4\lambda}}^{-1}=\frac{1}{1+4\lambda}\sum_{n=0}^\infty\frac{\Delta^n}{(1+4\lambda)^n}.
	\end{align*}
	Hence, $\norm{\mrm{B}\mrm{E}^2\Gamma}\leq\sum_{n=0}^\infty\frac{\norm{\mrm{B}\Delta^n\Gamma}}{(1+4\lambda)^{n+1}}$. For each term, we expand $\Delta^n$ to get
	\begin{align*}
		\norm{\mrm{B}\Delta^n\Gamma}&=\norm[\Big]{\expec_{\theta_1,\ldots,\theta_n\leftarrow\mu}\mrm{B}(I\otimes B_{\theta_1}\cdots B_{\theta_n}\otimes C_{\theta_1}\cdots C_{\theta_n})\Gamma}\\
		&=\norm[\Big]{\expec_{\theta_1,\ldots,\theta_n\leftarrow\mu}(I\otimes I\otimes C_{\theta_1}\cdots C_{\theta_n})\mrm{B}\Gamma(I\otimes B_{\theta_1}\cdots B_{\theta_n}\otimes I)}\\
		&\leq\expec_{\theta_1,\ldots,\theta_n\leftarrow\mu}\norm{I\otimes I\otimes C_{\theta_1}\cdots C_{\theta_n}}\norm{\mrm{B}\Gamma}\norm{I\otimes B_{\theta_1}\cdots B_{\theta_n}\otimes I}\\
		&=\norm{\mrm{B}\Gamma}.
	\end{align*}
	Thus, $\norm{\mrm{B}\mrm{E}^2\Gamma}\leqq\sum_{n=0}^\infty\frac{\norm{\mrm{B}\Gamma}}{(1+4\lambda)^{n+1}}=\frac{1}{1+4\lambda}\frac{1}{1-\frac{1}{1+4\lambda}}\norm{\mrm{B}\Gamma}=\frac{\norm{\mrm{B}\Gamma}}{4\lambda},$ giving $4\lambda^2\leq\norm{\mrm{B}\Gamma}$. To bound this overlap, we use \cref{lem:ool} to get $4\lambda^2\leq\sqrt{\norm{\Tr_A(\mrm{B}^2)}\norm{\Tr_A(\Gamma^2)}}$. Now we can relate this bound to the observable overlap:
	\begin{align*}
		\norm{\Tr_A(\mrm{B}^2)}=\norm[\Big]{\expec_{\theta,\theta'\leftarrow\mu}\Tr(A_\theta A_{\theta'})B_\theta B_{\theta'}}\leq\expec_{\theta,\theta'\leftarrow\mu}\abs*{\Tr(A_\theta A_{\theta'})}=c(\{A_\theta\}_\theta,\mu),
	\end{align*}
	and identically $\norm{\Tr_A(\Gamma^2)}\leq c(\{A_\theta\}_\theta,\mu)$. Hence, we find that $\lambda\leq\frac{1}{2}\sqrt{c(\{A_\theta\}_\theta,\mu)}$, as wanted.
\end{proof}

\section{Main result}

In this section, we shown unconditional uncloneable security of the Clifford scheme with messages of arbitrary length, using the results of the prior two sections. That is, for any polynomially-bounded functions $p(\lambda)$ and $q(\lambda)$ such that $p(\lambda)-q(\lambda)=\omega(\log(\lambda))$, \cref{thm:main} shows that the efficient QECM $\{\ttt{Q}_{\mc{C}_{p(\lambda)},q(\lambda)}\}_{\lambda}$ is strongly uncloneable-indistinguishable secure, and \cref{cor:unc} shows that it is strongly uncloneable secure as well.

\begin{theorem}\label{thm:main}
	Let $m\leq n$, and write $d=2^{n-m+1}$. Then, the cloning-distinguishing value is $$\mfk{cd}(\ttt{Q}_{\mc{C}_n,m})\leq\frac{1}{2}+\frac{d^{3/2}-1}{2(d^2-1)}\leq\frac{1}{2}+\frac{1}{2\sqrt{d}}.$$
\end{theorem}

A result of~\cite{BCR25} gives that the cloning-distinguishing value is lower-bounded by $\frac{1}{2}+O\parens*{\frac{1}{\sqrt{2^{n-m}}}}$, so this gives the optimal behaviour.

\begin{proof}
	Let $N=n-m+1$. Recall that $\ttt{G}_{\mc{C}_N,1}$ is the MoE game associated to $\ttt{Q}_{\mc{C}_N,1}$; we have that $\mfk{cd}(\ttt{Q}_{\mc{C}_N,1})=\mfk{c}(\ttt{Q}_{\mc{C}_N,1})\leq\mfk{w}(\ttt{G}_{\mc{C}_N,1})$. In the game, the observables consist of all the Paulis in $(\C^2)^{\otimes N}$ except for the identity, uniformly distributed. Consider the MoE game $\ttt{G}$ where the observables are all the Paulis in $(\C^2)^{\otimes N}$ with uniform question distribution. Since these form an orthogonal basis for $B((\C^2)^{\otimes N})$ with respect to the Hilbert-Schmidt inner product, we have by \cref{thm:ananth-sahai} that $\mfk{w}(\ttt{G})\leq\frac{1}{2}+\frac{1}{2}\sqrt{\frac{d}{d^2}}=\frac{1}{2}+\frac{1}{2\sqrt{d}}$. On the other hand, since $\ttt{G}$ has all the same observables as $\ttt{G}_{\mc{C}_N,1}$ except for the identity, corresponding to a measurement that always gives the same answer, we find that $\mfk{w}(\ttt{G})=\frac{d^2-1}{d^2}\mfk{w}(\ttt{G}_{\mc{C}_N,1})+\frac{1}{d^2}$. As such, we can bound
	\begin{align*}
		\mfk{w}(\ttt{G}_{\mc{C}_N,1})&\leq\frac{d^2}{d^2-1}\parens*{\frac{1}{2}+\frac{1}{2\sqrt{d}}-\frac{1}{d^2}}=\frac{1}{2}+\frac{d^{3/2}-1}{2(d^2-1)}.
	\end{align*}
	To finish, we use \cref{thm:one-to-many} to see that $\mfk{cd}(\ttt{Q}_{\mc{C}_n,m})\leq\mfk{cd}(\ttt{Q}_{\mc{C}_N,1})\leq\mfk{w}(\ttt{G}_{\mc{C}_N,1})\leq\frac{1}{2}+\frac{d^{3/2}-1}{2(d^2-1)}$, as wanted. To get the simpler but looser upper bound, note that
	\begin{align*}
		\frac{d^{3/2}-1}{2(d^2-1)}&=\frac{d^{2}-\sqrt{d}}{2(d^2-1)\sqrt{d}}\leq\frac{d^{2}-1}{2(d^2-1)\sqrt{d}}=\frac{1}{2\sqrt{d}}.\qedhere
		\end{align*}
\end{proof}

The bound above also extends to the cloning value of the QECM scheme.

\begin{corollary}\label{cor:unc}
	The cloning value of the above scheme is $$\mfk{c}(\ttt{Q}_{\mc{C}_n,m})\leq\frac{1}{2^m}+\frac{1}{2^{\frac{n-m+1}{2}}}.$$
\end{corollary}

\begin{proof}
	Suppose $\ttt{A}$ is a cloning attack against $\ttt{Q}_{\mc{C}_n,m}$. We construct a cloning-distinguishing attack $\ttt{A}'$ as follows. First, sample random messages $X_0$ and $X_1$. Then, Bob and Charlie run the cloning attack. For either player, if they receive output $x_0$, they output $0$, and otherwise they output $1$. If Alice samples $x_0$, the success probability is $\Pr[b_B=b_C=0|x_0]=\mfk{c}(\ttt{Q}_{\mc{C}_n,m},\ttt{A})$; if Alice samples $x_1$, the success probability is $$\Pr[b_B=b_C=1|x_1]=1-\Pr[m_B=m_C=x_0|x_1]\geq1-\Pr[m_B=x_0|x_1]\geq1-\frac{1}{2^m},$$ as $x_0$ and $x_1$ are sampled uniformly. As such, $\mfk{cd}(\ttt{Q}_{\mc{C}_n,m},\ttt{A}')\geq\frac{1}{2}\mfk{c}(\ttt{Q}_{\mc{C}_n,m},\ttt{A})+\frac{1}{2}\parens*{1-\frac{1}{2^m}}=\frac{1}{2}+\frac{1}{2}\parens*{\mfk{c}(\ttt{Q}_{\mc{C}_n,m},\ttt{A})-\frac{1}{2^m}}$. By \cref{thm:main}, $\mfk{cd}(\ttt{Q}_{\mc{C}_n,m},\ttt{A}')\leq\frac{1}{2}+\frac{1}{2\sqrt{2^{n-m+1}}}$, which implies $\mfk{c}(\ttt{Q}_{\mc{C}_n,m},\ttt{A})\leq\frac{1}{2^m}+\frac{1}{\sqrt{2^{n-m+1}}}$.
\end{proof}

\small
\bibliography{bibtex/bib/full.bib,bibtex/bib/quantum.bib,bibtex/quantum_new.bib}

\normalsize

\end{document}